\documentclass{article}
\usepackage{a4wide}
\usepackage{latexsym}
\usepackage{amsmath}
\usepackage{amssymb}
\usepackage{amsthm}
\usepackage{ifthen}
\usepackage{mathrsfs}
\usepackage{graphics}
\usepackage{color}
\usepackage{algpseudocode}
\usepackage[ruled,vlined,linesnumbered]{algorithm2e}
\usepackage{tikz}
\usetikzlibrary{calc}
\usepackage{amsmath}
\usetikzlibrary{decorations.pathreplacing}
\usepackage{todonotes}

\newcommand{\nc}{\newcommand}
\nc{\rnc}{\renewcommand} \nc{\nev}{\newenvironment}

\nc{\W}[1]{\text{$\textup{W}[#1]$}}
\nc{\FPT}{\textup{FPT}}
\nc{\fpt}{\textup{fpt}}

\newlength{\probwidth}
\nc{\prob}[3][9]{
\begin{center}
  \normalfont\fbox{
   \begin{tabular}[t]{
     rp{#1cm}}\textit{Instance:}&#2. \\
     \textit{Problem:}&#3
   \end{tabular}}
\end{center}}

\nc{\pprob}[4][9]{
\begin{center}
   \normalfont\fbox{
    \begin{tabular}[t]{
     rp{#1cm}}\textit{Instance:}&#2. \\
     \textit{Parameter:}&#3. \\
     \textit{Problem:}&#4
   \end{tabular}}
\end{center}}

\nc{\nprob}[4][9]{
\begin{center}
  \normalfont\fbox{

\addtolength{\probwidth}{#1cm}\parbox{\probwidth}{\textsc{#2}\\\hspace*{1.5em}
     \begin{tabular}[t]{
      rp{#1cm}}\textit{Instance:}&#3. \\
      \textit{Problem:}&#4
     \end{tabular}}}
\end{center}}

\nc{\npprob}[5][9]{
\begin{center}
  \normalfont\fbox{

\addtolength{\probwidth}{#1cm}\parbox{\probwidth}{\textsc{#2}\\\hspace*{1.5em}
    \begin{tabular}[t]{
     rp{#1cm}}\textit{Instance:}&#3. \\
     \textit{Parameter:}&#4. \\
     \textit{Problem:}&#5
    \end{tabular}}}
\end{center}}

\nc{\nppxrob}[5][9]{ \normalfont\fbox{

\addtolength{\probwidth}{#1cm}\parbox{\probwidth}{\textsc{#2}\\\hspace*{1.5em}
   \begin{tabular}[t]{
    rp{#1cm}}\textit{Instance:}&#3. \\
    \textit{Parameter:}&#4. \\
    \textit{Problem:}&#5
   \end{tabular}}}}

\nc{\nppprob}[5][4]{
\begin{center}
  \normalfont\fbox{

\addtolength{\probwidth}{#1cm}\parbox{\probwidth}{\textsc{#2}\\\hspace*{1.5em}
    \begin{tabular}[t]{
     rp{#1cm}}\textit{Instance:}&#3. \\
     \textit{Parameter:}&#4. \\
     \textit{Problem:}&#5
    \end{tabular}}}
\end{center}}

\nc{\dotcup}{\;\dot\cup\;}

\nc{\noptprob}[6][9]{
\begin{center}
  \normalfont\fbox{

\addtolength{\probwidth}{#1cm}\parbox{\probwidth}{\textsc{#2}\\\hspace*{0em}
    \begin{tabular}[t]{
     rp{10cm}}\textit{Instance:}&#3. \\
     \textit{Solution:}&#4. \\
     \textit{Cost:}&#5. \\
     \textit{Goal:}&#6.
    \end{tabular}}}
\end{center}}

\newtheorem{theo}{Theorem}[section]
\newtheorem{lemma}[theo]{Lemma}

\newtheorem{definition}[theo]{Definition}

\newtheorem{que}[theo]{Question}

\rnc{\c}{\mathbf{{c}}} 
\rnc{\a}{\mathbf{{a}}}
\rnc{\b}{\mathbf{{b}}} 
\rnc{\i}{\mathbf{{i}}}
\rnc{\j}{\mathbf{{j}}} 
\rnc{\u}{\mathbf{{u}}}
\rnc{\v}{\mathbf{{v}}}
\nc{\B}{\mathbf{{B}}}
\nc{\s}{\mathbf{{s}}}
\rnc{\S}{\mathbf{{S}}}
\nc{\C}{\mathbf{{C}}}
\nc{\U}{\mathbf{{U}}}
\nc{\E}{\mathbf{{E}}}
\nc{\G}{\mathbf{{G}}}
\nc{\D}{\mathbf{{D}}}
\nc{\unsat}{\mathbf{{UNSAT}}}

\author{Bingkai Lin \\  Nanjing University, China \\ \texttt{lin@nju.edu.cn} }
\title{Constant Approximating $k$-Clique is W[1]-hard}

\begin{document}

\maketitle

\begin{abstract}
For every graph $G$, let $\omega(G)$ be the largest size of complete subgraph in $G$. This paper presents a simple  algorithm which, on input a graph $G$, a positive integer $k$ and a small constant $\epsilon>0$, outputs a graph $G'$ and an integer $k'$ in $2^{\Theta(k^5)}\cdot |G|^{O(1)}$-time such that (1) $k'\le 2^{\Theta(k^5)}$, (2) if $\omega(G)\ge k$, then $\omega(G')\ge k'$, (3) if $\omega(G)<k$, then $\omega(G')< (1-\epsilon)k'$. This implies that  no $f(k)\cdot |G|^{O(1)}$-time algorithm  can distinguish between the cases $\omega(G)\ge k$ and $\omega(G)<k/c$ for any constant $c\ge 1$ and  computable function $f$, unless $FPT= W[1]$. 
\end{abstract}
\section{Introduction}
Given a simple graph $G$ and a positive integer $k$, the task of $k$-Clique problem  is to decide whether $\omega(G)\ge k$.  In  parameterized complexity~\cite{dowfel99,flugro06}, the $k$-Clique problem with $k$ as its parameter is a canonical $W[1]$-complete problem~\cite{downey1995fixed}. Unless $W[1]= FPT$, it has  no $f(k)\cdot |G|^{O(1)}$-time algorithm (FPT-algorithm) for any computable function $f : \mathbb{N}\to\mathbb{N}$. This problem  has been used as a starting point in many reductions and thus plays a fundamental role in the area of parameterized complexity. Yet, it is still not known whether constant approximating $k$-Clique is also $W[1]$-hard. More precisely, we consider the following  question:
\begin{que}\label{que:w1appclique}
Is there any algorithm which, on input a graph $G$ and a positive integer $k$, outputs a new graph $G'$ and a positive integer $k'$ in $f(k)\cdot |G|^{O(1)}$-time for some computable function $f:\mathbb{N}\to\mathbb{N}$ such that,
\begin{itemize}
\item $k'=g(k)$ for some computable function $g:\mathbb{N}\to\mathbb{N}$,
\item if $\omega(G)\ge k$, then $\omega(G')\ge k'$,
\item if $\omega(G)< k$, then $\omega(G')< k'/2$?
\end{itemize} 
\end{que}
The question above is motivated by the study of FPT-approximation algorithms for $k$-Clique. For any $c\ge 1$, we say an algorithm is a $c$-FPT-approximation algorithm for $k$-Clique  if on input a graph $G$  it outputs a clique of size $\omega(G)/c$  in $G$ in $f(\omega(G))\cdot |G|^{O(1)}$-time for some computable function $f : \mathbb{N}\to\mathbb{N}$. Whether there exists such an algorithm has been repeatedly raised in the literature~\cite{marx08,fellowsdata,chegro07,downey2013fundamentals}.    Previous results for FPT-inapproximability of $k$-Clique are under strong assumptions which already have a gap~\cite{chalermsook2017gap,bonesc13}. Proving such results based on standard assumptions is an interesting and important open question. It is wildly believed that the technique needed to resolve this question is closely related to a PCP-theorem for parameterized complexity~\cite{chegro07}:
\begin{quote}
Non-approximability results in the classical framework were proved for the CLIQUE problem using the PCP-theorem, so it might be necessary to obtain a parameterized version of the PCP-theorem to solve these questions.
\end{quote}

This paper gives the first positive answer to Question~\ref{que:w1appclique}. An immediate corollary of our result is the non-existence of FPT-approximation algorithm for the $k$-Clique problem under the standard parameterized complexity hypothesis  $W[1]\neq FPT$.
\begin{theo}
Assuming that $k$-Clique has no FPT-algorithm,  there is no  FPT-algorithm that can approximate $k$-Clique to any constant.
\end{theo}
The  main contribution of this paper is to show how to create a constant gap for $\omega(G)$ from a $W[1]$-hard problem with no gap.

\subsection{Overview of the reduction}
Let us illustrate the idea  using  a toy $k$-Vector-Sum problem and the Walsh-Hadamard  code. Throughout this paper, we will work on finite field  $\mathbb{F}$ with characteristic $2$. For any $\vec{v}_1,\ldots,\vec{v}_k\in\mathbb{F}^d$ and $\vec{a}_1,\ldots,\vec{a}_k\in\mathbb{F}^d$, let $H(\vec{v}_1,\ldots,\vec{v}_k)_{\vec{a}_1,\ldots,\vec{a}_k}=\sum_{i\in[k]} \vec{a}_i\cdot\vec{v}_i$ be the Walsh-Hadamard  code of $\vec{v}_1,\ldots,\vec{v}_k$. Here $\vec{a}_i\cdot\vec{v}_i$ denotes the dot product of vectors $\vec{a}_i$ and $\vec{v}_i$. Given $k$ vector $\vec{v}_1,\ldots,\vec{v}_k\in\mathbb{F}^d$ and a target vector $\vec{t}\in\mathbb{F}^d$, we  want to test whether $\sum_{i\in[k]}\vec{v}_i=\vec{t}$. For that sake, we construct a constraint  satisfaction problem (CSP) on variables $\{x_{\vec{a}_1,\ldots,\vec{a}_k}:\vec{a}_1,\ldots,\vec{a}_k\in\mathbb{F}^d\}$. Any assignment to these variables can be seen as a long vector   $\vec{x}\in\mathbb{F}^{kd}$, which is supposed to be the Walsh-Hadamard  code of  $\vec{v}_1,\ldots,\vec{v}_k$. Then we do the following tests.
\begin{itemize}
\item[(T1)]   To ensure that the vector $\vec{x}\in\mathbb{F}^{kd}$ is a Walsh-Hadamard  code   of some vector $\vec{v}_1,\ldots,\vec{v}_k\in\mathbb{F}^d$, we  check if  $\vec{x}_{\vec{a}_1+\vec{b}_1,\ldots,\vec{a}_k+\vec{b}_1}=\vec{x}_{\vec{a}_1,\ldots,\vec{a}_k}+\vec{x}_{\vec{b}_1,\ldots,\vec{b}_k}$ for  $\vec{a}_1,\ldots,\vec{a}_k,\vec{b}_1,\ldots,\vec{b}_k\in\mathbb{F}^d$.
\item[(T2)]  Check if $\vec{x}_{\vec{a}_1,\ldots,\vec{a}_i+\vec{a},\ldots,\vec{a}_k}-\vec{x}_{\vec{a}_1,\ldots,\vec{a}_k}=\vec{v}_i\cdot\vec{a}$ for   $i\in[k]$, $\vec{a}\in\mathbb{F}^d$ and  $\vec{a}_1,\ldots,\vec{a}_k\in\mathbb{F}^d$.
\item[(T3)] Check if $\vec{x}_{\vec{a}_1+\vec{a},\ldots,\vec{a}_k+\vec{a}}-\vec{x}_{\vec{a}_1,\ldots,\vec{a}_k}=\vec{a}\cdot\vec{t}$ for  $\vec{a},\vec{a}_1,\ldots,\vec{a}_k\in\mathbb{F}^d$.
\end{itemize}
By the linearity testing~\cite{blum1993self}, for any assignment $\vec{x}\in \mathbb{F}^{kd}$, if $\vec{x}$ satisfies $\delta$-fractions of  constraints in (T1),  then $\vec{x}$ is at most $(1-\delta)$-far from the Walsh-Hadamard  code of some vectors $\vec{u}_1,\ldots,\vec{u}_k\in\mathbb{F}^d$. That is, for a $\delta$-fraction of $(\vec{a}_1,\ldots,\vec{a}_k)\in\mathbb{F}^{dk}$, $\vec{x}_{\vec{a}_1,\ldots,\vec{a}_k}=H(\vec{u}_1,\ldots,\vec{u}_k)_{\vec{a}_1,\ldots,\vec{a}_k}$. Next, if $\vec{u}_1+\cdots+\vec{u}_k\neq\vec{t}$, then for a $(1-1/|\mathbb{F}|)$-fraction of $\vec{a}\in\mathbb{F}^d$, $\vec{a}\cdot (\vec{u}_1+\cdots+\vec{u}_k)\neq\vec{a}\cdot\vec{t}$. If for some $i\in[k]$, $\vec{u}_i\neq\vec{v}_i$, then  for a $(1-1/|\mathbb{F}|)$-fraction of $\vec{a}\in\mathbb{F}^d$, $\vec{a}\cdot \vec{u}_i\neq\vec{a}\cdot\vec{v}_i$. We conclude that, if $\vec{v}_1+\cdots+\vec{v}_k=\vec{t}$, then there exists a vector $\vec{x}$ satisfying all the constraints. If $\vec{v}_1+\cdots+\vec{v}_k\neq\vec{t}$, then for any assignment $\vec{x}$, one of the following must hold:
\begin{itemize}
\item  a constant fraction of constraints in (T1) are not satisfied.
\item for some $i\in[k]$, a constant fraction of constraints in (T2) with respect to $i$ are not satisfied.
\item a constant fraction of constraints in (T3) are not satisfied.
\end{itemize}
%Then we construct a graph 

Then we construct a graph $G$ from the CSP instance using a modified FGLSS-reduction~\cite{feige1996interactive}. The vertex set of $G$ consists of two parts $A$ and $B$.

The vertices of $A$ are corresponding to assignments to variables $\{x_{\vec{a}_1,\ldots,\vec{a}_k}:\vec{a}_1,\ldots,\vec{a}_k\in\mathbb{F}^d\}$ of the CSP instance. Two vertices in $A$ are adjacent unless they are corresponding to  assignments that are not consistent, i.e., assigning different values to the same variable, or there exists a (T2) or (T3) constraint between them and they do not satisfy that constraint. Note that there are $|\mathbb{F}|^{1+dk}$ vertices in $A$. They can be partitioned into $|\mathbb{F}|^{dk}$ independent sets of size $|\mathbb{F}|$.

%For constraints in (T1), we use FGLSS-reduction to construct a graph $G_{T1}$ whose 
The vertices in $B$ are corresponding to assignments to three variables satisfying constraints in (T1). Two vertices in $B$ are adjacent unless they are corresponding to  inconsistent assignments. Note that the vertex set of $B$ can be partitioned into $|\mathbb{F}|^{2dk}$ disjoint subsets, each forms an independent set of size $|\mathbb{F}|^2$.

We add an edge between any two vertices in $A$ and $B$ if  the assignments corresponding to these vertice are consistent.
Since  the sizes of $A$ and $B$ are not balanced, we assign to each vertex in $A$  weight $|\mathbb{F}|^{kd}$ and each vertex in $B$ weight $1$.

If $\vec{v}_1+\cdots+\vec{v}_k=\vec{t}$, then $G$ contains a clique of weight $2|\mathbb{F}|^{2kd}$. This clique consists of $|\mathbb{F}|^{kd}$ vertices from $A$ and $|\mathbb{F}|^{2kd}$ vertices from $B$. 
If $\vec{v}_1+\cdots+\vec{v}_k\neq\vec{t}$, then for any clique $X$ in $G$, one of the following must hold:
\begin{itemize}
\item[(i)] $|X\cap B|\le (1-\epsilon)|\mathbb{F}|^{2kd}$,
\item[(ii)] there exists $\vec{a}\in\mathbb{F}^d$ and $i\in[k]$ such that the variable set $V$ can be partitioned into two disjoint sets $V=V_0\cup V_1$ 
with $V_1=\{x_{\vec{a}_1,\ldots,\vec{a}_i+\vec{a},\ldots,\vec{a}_k} : x_{\vec{a}_1,\ldots,\vec{a}_k}\in V_0\}$. The assignment corresponding to $X$ satisfies only a $(1-\epsilon)$-fraction of (T2) constraints between $V_0$ and $V_1$,
\item[(iii)] there exists $\vec{a}\in\mathbb{F}^d$ such that the variable set $V$ can be partitioned into two disjoint sets $V=V_0\cup V_1$ 
with $V_1=\{x_{\vec{a}_1+\vec{a},\ldots,\vec{a}_k+\vec{a}} : x_{\vec{a}_1,\ldots,\vec{a}_k}\in V_0\}$. The assignment corresponding to $X$ satisfies only a $(1-\epsilon)$-fraction of (T3) constraints between $V_0$ and $V_1$. 
%with $V_1=\{(\vec{a}_1+\vec{a},\ldots,\vec{a}_k+\vec{a}) : (\vec{a}_1,\ldots,\vec{a}_k)\in V_0\}$. $|X\cap A|\le (1-\epsilon)|\mathbb{F}|^{kd}$
\end{itemize}
Either (ii) or (iii) implies that $|X\cap A|\le [1/2+(1-\epsilon)1/2]|\mathbb{F}^{kd}|$. So, in summary, when $\vec{v}_1+\cdots+\vec{v}_k\neq\vec{t}$, we have for every clique of $G$, either $|X\cap B|\le (1-\epsilon)|\mathbb{F}|^{2kd}$ or $|X\cap A|\le (1-\epsilon/2)|\mathbb{F}^{kd}|$. In both cases, every clique in $G$ has at most $(2-\epsilon/2)|\mathbb{F}|^{2kd}$ weight.
%and a target vector $\vec{t}\in\mathbb{F}^d$.

\bigskip

There are two problems needed to be solved:
\begin{itemize}
\item[(P1)] In the real $k$-Vector-Sum problem, we are given $k$ sets $V_1,\ldots,V_k$ of vectors  instead of $k$ vectors. How to test whether $\vec{x}_{\vec{a}_1,\ldots,\vec{a}_i+\vec{a},\ldots,\vec{a}_k}-\vec{x}_{\vec{a}_1,\ldots,\vec{a}_k}=\vec{v}_i\cdot\vec{a}$ for the same vector $\vec{v}_i\in V_i$? In some bad scenario, it is possible that for each  $\vec{a}$, there exists a vector $\vec{v}_{\vec{a}}\in V_i$ such that $\vec{x}_{\vec{a}_1,\ldots,\vec{a}_i+\vec{a},\ldots,\vec{a}_k}-\vec{x}_{\vec{a}_1,\ldots,\vec{a}_k}=\vec{v}_{\vec{a}}\cdot\vec{a}$.
% and $\vec{x}_{\vec{a}_1,\ldots,\vec{a}_i+\vec{a}',\ldots,\vec{a}_k}-\vec{x}_{\vec{a}_1,\ldots,\vec{a}_k}=\vec{v}_i'\cdot\vec{a}'$
\item[(P2)] The $k$-Vector-Sum problem is $W[1]$-hard when $d=k^{\Omega(1)}\log n$. Applying our reduction directly would take at least $|\mathbb{F}|^{kd}\ge n^{k^{\Omega(1)}}$ time, which we cannot afford. 
\end{itemize}
We handle these problems by   sampling $\ell$ matrices $A_1,\ldots,A_\ell\in\mathbb{F}^{h\times d}$ with $h=k^2$, and  replacing each vector $\vec{v}\in V_i$ by an $\ell$-tuple $(A_1 \vec{v},\ldots,A_\ell\vec{v})$ of $h$-dimension vectors.  For every $\vec{a}_1,\ldots,\vec{a}_k\in\mathbb{F}^{h}$, the value of the variable $x_{\vec{a}_1,\ldots,\vec{a}_k}$ becomes a  vector in $\mathbb{F}^\ell$ instead of just an element of $\mathbb{F}$. The number of variables become $|\mathbb{F}|^{kh}$ instead of $|\mathbb{F}|^{kd}$.  Since  $|\mathbb{F}|=O(1)$, the reduction can be done in FPT-time if $\ell\le O(\log n+h)$.
To ensure that the constraints in (T2) still work, we show that when $\ell\ge \Omega(\log n+h)$, with high probability, for all $\vec{a}\in\mathbb{F}^h$ and distinct $\vec{v},\vec{u}\in V_i$, $(\vec{a}^T A_1\vec{v},\ldots,\vec{a}^T A_\ell\vec{v})\neq (\vec{a}^T A_1\vec{u},\ldots,\vec{a}^T A_\ell \vec{u})$ and the bad scenario in (P1) will not occur. For constraints in (T3), we show that with high probability, for all distinct $\vec{t},\vec{t}'\in\mathbb{F}^d$,  $( A_1\vec{t},\ldots, A_\ell\vec{t})\neq (A_1\vec{t}',\ldots, A_\ell\vec{t}')$.
%Our technique contribution, alphabet reduction for large alphabet.

\subsection{Related work}
The $k$-Clique problem is one of the first known NP-hard problems in \cite{kar72}. It was  showned that approximating $k$-Clique to a factor of  $n^{1-\epsilon}$  is also NP-hard after a  long line of research~\cite{feige1996interactive,bellare1993efficient,bellare1994improved,goldwasser1998introduction,feige2000two,hastad1996clique,zuckerman2006linear}. As pointed out in~\cite{chen2019constant}, the classical inapproximability results of $k$-Clique inevitably produce instances with large $\omega(G)$. However, in parameterized complexity, we consider instances with  small $\omega(G)$ which does not depend on the size of $G$. In order to show that $k$-Clique is still hard to approximate when $k$ is small, a natural idea is to use the PCP-theorem~\cite{arora1998proof,arora1998probabilistic} to obtain a gap for the SAT problem and then use the method of compressing to reduce the optimum solution size~\cite{hajkhakor13}. Unfortunately, since the PCP-theorem causes a polylogarithmic blow-up  in the size of SAT instance, this approach cannot rule out FPT-approximation for $k$-Clique. 

To circumvent this, researches  used stronger hypothesis to obtain a gap for the SAT problem. Bonnet et al.~\cite{bonesc13} used ETH~\cite{impagliazzo2001problems} and the linear PCP conjecture to show  constant FPT-inapproximability of $k$-Clique. Assuming Gap-ETH~\cite{dinur2016mildly,manurangsi2016birthday}, it was shown that there is no $o(k)$-FPT-approximation for $k$-Clique~\cite{chalermsook2017gap}.

In \cite{pih},  a weaker conjecture called Parameterized Inapproximability Hypothesis (PIH) was postulated. PIH states that binary CSP parameterized by the number of variables has no constant FPT-approximation. It is easy to see that PIH implies  $k$-Clique has no constant FPT-approximation. Interestingly, it is not known if the hardness of approximation of $k$-Clique implies PIH~\cite{feldmann2020survey}.

Assuming a conjecture called DEG-2-SAT, Khot and Shinkar~\cite{khot2016hardness} used a different approach to rule out FPT-approximation for $k$-Clique. Although the conjecture turned out to be false~\cite{kayal2014solvability}, their work is still  inspiring. The idea of  multiplying  the input instance with  matrices in our reduction is from their paper.

In recent years, several gap-creating techniques  have been successfully used to  show FPT  inapproximabilities~\cite{lin2018parameterized,karthik2017parameterized,lin2019simple,wlodarczyk2020parameterized,karthik2020hardness}. 
We refer the reader to~\cite{feldmann2020survey} for a survey of these results.
 
\section{Preliminaries}
For every vector $\vec{v}=(v_1,\ldots,v_d)\in\mathbb{F}^d$ and $i\in[d]$, let $\vec{v}[i]=v_i$.
For every $\vec{x},\vec{y}\in\mathbb{F}^n$, we use $\vec{x}\cdot\vec{y}=\sum_{i\in[n]}\vec{x}[i]\vec{y}[i]$ to denote the dot product of $\vec{x}$ and $\vec{y}$. For $X\subseteq \mathbb{F}^d$ and $\vec{v}\in\mathbb{F}^d$, let $X+\vec{v}=\{\vec{x}+\vec{v} : \vec{x}\in X\}$. For $\vec{a}\in\mathbb{F}^n$ and $\vec{b}\in\mathbb{F}^m$, let $\vec{a}\circ\vec{b}$ be the  the result of concatenating $\vec{a}$ and $\vec{b}$.
For any $d n$ vectors $\vec{v}=(v_1,\ldots,v_{dn})\in\mathbb{F}^{dn}$ and $\vec{a}=(a_1,\ldots,a_d)\in\mathbb{F}^d$, let
\[
F(\vec{a},\vec{v})=(\sum_{i\in [d]}a_iv_i,\ldots,\sum_{i\in [d]}a_iv_{i+jd},\ldots,\sum_{i\in [d]}a_iv_{i+(n-1)d})\in\mathbb{F}^n.
\]

For all $n\in\mathbb{N}$, let $\vec{0}_n$ and $\vec{1}_n$ be the $n$-dimension all-zero vector and all-one vector respectively.
\begin{definition}[Distance]
For every $d\in\mathbb{N}$ and $\vec{x},\vec{y}\in \mathbb{F}^d$, let 
\[
dist(\vec{x},\vec{y})=\frac{|\{i\in[d] : \vec{x}[i]\neq \vec{y}[i]\}|}{d}.
\]
For ease of notation, let $dist(\vec{x})=dist(\vec{x},\vec{0}_d)$.
\end{definition}
Let $G$ and $H$ be two groups and $+$ be the group operator.
For any $\delta\in[0,1]$ and $f,g: G\to H$, we say $f$  is $\delta$-far from $g$ if $\Pr_{x}[f(x)\neq g(x)]= \delta$. A function $f:G\to H$ is a homomorphism if $f(x)+f(y)=f(x+y)$ for all $x,y\in G$.
%$dist(\vec{x},\vec{x}')\le 1-\delta$. 
The follow theorem is from \cite{blum1993self,goldreich2016lecture}.
\begin{theo}[Linearity test]
If a function $f: G\to H$ satisfies
$\Pr_{{x},{y}}[f(x)+f(y)= f({x}+{y})]\ge (1-\delta/2)$ for some small $\delta$, then there exists a  homomorphism function $g : G\to H$ such that $f$ is at most $\delta$-far from $g$, i.e.,
$\Pr_{x}[f(x)= g(x)]\ge (1-\delta)$.
%is $\delta$-far from any homomorphism function  for some small $\delta$, then  $\Pr_{{x},{y}}[f(x)+f(y)\neq f({x}+{y})]\ge \delta/2$.
%for small $\epsilon>0$, then $f$  is $2\epsilon$-close to some linear function $g(x_1,\ldots,x_d)=\Sigma_i c_ix_i$.
\end{theo}

\begin{definition}[Constraint Satisfaction Problem (CSP)]
Given an alphabet $\Gamma$, an instance of  constraint satisfaction problem contains a set of variables $V=\{v_1,\ldots,v_k\}$  and  constraints $\{C_1,C_2,\ldots,C_m\}$. For every $i\in [m]$, 
 $C_i=(\vec{s}_i,R_i)$, where $\vec{s}_i=(v_{j_1},\ldots,v_{j_{\ell_i}})$ is an $\ell_i$-tuple of variables for some $\ell_i\in[k]$ and $R_i\subseteq\Gamma^{\ell_i}$.
%there exists $i_1,\ldots,i_\ell\in[k]$ 
% between every pair of adjacent vertices in $V(G)$, where $\Sigma$ is the alphabet for the vertices of $G$. 
 The goal is to find  an assignment $\sigma : V\to \Gamma$  such that
\begin{itemize}
\item for all $i\in[m]$, $\sigma(\vec{s}_i)\in R_i$.
\end{itemize}
\end{definition}

In parameterized complexity, the hypothesis of $W[1]\neq FPT$ states that no algorithm  can, on input a graph $G$ and a positive integer $k$, decide whether $\omega(G)\ge k$ in $f(k)\cdot |G|^{O(1)}$ time for any computable function $f:\mathbb{N}\to\mathbb{N}$. 
A parameterized problem $L$ is $W[1]$-hard if there is a reduction  from $k$-Clique to this problem such that for every instance $(G,k)$ of $k$-Clique,  the reduction outputs an instance $(x,k')$ of $L$ in $f(k)\cdot |G|^{O(1)}$-time for some computable function $f:\mathbb{N}\to\mathbb{N}$ and:
\begin{itemize}
\item $(G,k)$ is a yes-instance of $k$-Clique if and only if $(x,k')$ is a yes-instance of $L$,
\item $k'\le g(k)$ for some computable function $g :\mathbb{N}\to\mathbb{N}$.
\end{itemize}
Obviously, if a parameterized problem $L$ is $W[1]$-hard, then   no $f(k)\cdot |x|^{O(1)}$-time algorithm   can decide whether $(x,k)$ is a yes-instance of $L$ unless $W[1]=FPT$.  We say approximating $k$-Clique to a factor of $c$ is $W[1]$-hard if the existence of $f(k)\cdot |G|^{O(1)}$-time algorithm  that can distinguish $\omega(G)\ge k$ and $\omega(G)<k/c$ would imply $W[1]=FPT$.
%there is an fpt-reduction from a $W[1]$-hard problem to the gap-version of $k$-Clique problem.

\begin{definition}[$k$-Vector-Sum]
Given $k$ sets $V_1,\ldots,V_k$ of vectors and a target vector $\vec{t}$ in $\mathbb{F}^{m}$, the goal of $k$-vector-sum problem is to decide whether there exist $\vec{v}_1\in V_1,\ldots,\vec{v}_k\in V_k$ such that
\[
\sum_{i\in[k]}\vec{v}_i=\vec{t}.
\]
\end{definition}
The $W[1]$-hardness of $k$-Vector-Sum was proved in  \cite{abboud2013exact}. For the convenience of the reader, we include a proof in the Appendix.
\begin{theo}\label{thm:kvec}
$k$-Vector-Sum with $\mathbb{F}=\mathbb{F}_2$  and $m=\Theta(k^2\log n)$ is $W[1]$-hard parameterized by $k$.
\end{theo}

\section{Gap-reduction from $k$-Vector-Sum to $k$-Clique}

Given an instance $(V_1,V_2,\ldots,V_k,\vec{t})$ of $k$-Vector-Sum over $\mathbb{F}^{m}$. Let  $V=\bigcup_{i\in[k]}V_i$, $n=|V|$ and $h=\Theta(k^2)$. By Theorem~\ref{thm:kvec}, we can assume that $m=h\log n$. Let $\mathbb{F}$ be a finite field  with $|\mathbb{F}|=4$. Vectors in the hardness instances from Theorem~\ref{thm:kvec} can still be treated as vectors in $\mathbb{F}^m$. Since  $V$ only contains  vectors in $\{0,1\}^m$, we have 
\begin{equation}\label{eq:ucv}
\vec{v}\neq c\vec{u} \text{ for any  distinct $\vec{v},\vec{u}\in V$ and  nonzero $c\in\mathbb{F}$}.
\end{equation}

For any  $\ell\in\mathbb{N}$, select $\ell$ matrices $A_1,A_2,\ldots,A_\ell\in \mathbb{F}^{h\times m}$ randomly and independently. For every  $\vec{v}\in \mathbb{F}^m$, let  
\[
g(\vec{v})=(A_1\vec{v},\cdots,A_\ell\vec{v})\in\mathbb{F}^{h\ell}.
\]
For every   vector $\vec{\alpha}\in\mathbb{F}^h$ and $\vec{v}\in \mathbb{F}^m$, let  
\[
f(\vec{\alpha},\vec{v})=(\vec{\alpha}^TA_1\vec{v},\cdots,\vec{\alpha}^T A_\ell\vec{v})\in\mathbb{F}^\ell.
\]
Recall that for any $d\cdot n$ vectors $\vec{v}=(v_1,\ldots,v_{dn})\in\mathbb{F}^{dn}$ and $\vec{a}=(a_1,\ldots,a_d)\in\mathbb{F}^d$, 
\[
F(\vec{a},\vec{v})=(\sum_{i\in [d]}a_iv_i,\ldots,\sum_{i\in [d]}a_iv_{i+jd},\ldots,\sum_{i\in [d]}a_iv_{i+(n-1)d})\in\mathbb{F}^n.
\]
It follows that $f(\vec{\alpha},\vec{v})=F(\vec{\alpha},g(\vec{v}))$. Note that for every $\vec{v}\in\mathbb{F}^{dn}$, $F(\cdot,\vec{v}) : \mathbb{F}^{d}\to\mathbb{F}^n$ is a homomorphism from $\mathbb{F}^{d}$ to $\mathbb{F}^{n}$. Every homomorphism $f$ from $\mathbb{F}^{d}$ to $\mathbb{F}^{n}$ is also a bitwise linear function, so it can be written as $f(\cdot)=F(\cdot,\vec{v})$ for some $\vec{v}\in\mathbb{F}^{dn}$.

\begin{lemma}\label{lem:difft}
If $1/10 > (1/|\mathbb{F}|)^{\ell h}\cdot 2^m$, then with  probability at least $9/10$, for all nonzero vector $\vec{v}\in\mathbb{F}^m$, $g(\vec{v}))\neq \vec{0}_{\ell h}$.
\end{lemma}
\begin{proof}
For any nonzero vector $\vec{v}\in\mathbb{F}^m$, $\Pr[A_i\vec{v}=\vec{0}_h]=(1/|\mathbb{F}|)^h$.  
\[
\Pr[g(\vec{v})=\vec{0}_{\ell h}]=\prod_{i\in[\ell]}\Pr[A_i\vec{v}=\vec{0}_h]= (1/|\mathbb{F}|)^{\ell h}.
\]
With probability at least
\[
1-(1/|\mathbb{F}|)^{\ell h}\cdot 2^m\ge 9/10,
\]
$g(\vec{v})\neq\vec{0}_{\ell h}$ for all nonzero $\vec{v}\in\mathbb{F}^m$.
\end{proof}

\begin{lemma}\label{lem:randmatrix}
If $A\in\mathbb{F}^{h\times m}$ is a random  matrix, then
 for any  nonzero vectors $\vec{b},\vec{c}\in\mathbb{F}^h$ and distinct $\vec{v},\vec{u}\in\mathbb{F}^m$ with $\vec{v}\neq a\vec{u}$ for any $a\in\mathbb{F}\setminus\{0\}$, 
%\begin{equation}
%\Pr[\vec{a}^TA\vec{v}=\vec{a}^TA\vec{u}]\le 1/|\mathbb{F}|.
%\end{equation}
\begin{equation}
\Pr[\vec{b}^TA\vec{v}=\vec{c}^TA\vec{u}]= 1/|\mathbb{F}|.
\end{equation}
\end{lemma}
\begin{proof}
%Suppose $\vec{b}=(b_1,\ldots,b_h)$
Let $B$ be an $h\times m$ matrix with $B_{ij}=\vec{b}[i]\cdot\vec{v}[j]$. Let $C$ be an $h\times m$ matrix with $C_{ij}=\vec{c}[i]\cdot\vec{u}[j]$. We can treat $A,B$ as a vector of length $hm$ and use $A\cdot B$ denote their dot product. It follows that $\vec{b}^TA\vec{v}=B\cdot A$ and $\vec{c}^TA\vec{u}=C\cdot A$. 
Since $\vec{v}\neq a\vec{u}$ for any  nonzero $a\in\mathbb{F}$, we have $B-C$ is not a zero matrix. Therefore,
\[
\Pr[\vec{b}^TA\vec{v}=\vec{c}^TA\vec{u}]=\Pr[B\cdot A=C\cdot A]=\Pr[(B-C)\cdot A=0]=1/|\mathbb{F}|.
\]
%$
%\Pr[\vec{\alpha}^TA\vec{v}=\vec{\alpha}^TA\vec{u}]
%=\Pr[\vec{\alpha}^TA(\vec{v}-\vec{u})=0]
%=\Pr[ A(\vec{v}-\vec{u})\in\vec{\alpha}^{\perp}]
%\le 1/|\mathbb{F}|.
%$
\end{proof}

\begin{lemma}\label{lem:alphav}
%Select $\ell$ matrices $A_1,A_2,\ldots,A_\ell\in \mathbb{F}^{h\times m}$ randomly and independently.  For any  nonzero vector $\vec{\alpha}\in\mathbb{F}^h$ and $\vec{v}\in V$, let  
%\[
%f(\vec{\alpha},\vec{v})=(\vec{\alpha}A_1\vec{v}^T,\cdots,\vec{\alpha} A_\ell\vec{v}^T)\in\mathbb{F}^\ell.
%\]
If $|V|^2\cdot |\mathbb{F}|^h\cdot (1/|\mathbb{F}|)^\ell<1/10$, then with  probability at least $9/10$, $f(\vec{\alpha},\vec{v})\neq f(\vec{\alpha},\vec{u})$ for any distinct $\vec{v},\vec{u}\in V$ and nonzero $\vec{\alpha}\in\mathbb{F}^h$.
\end{lemma}

\begin{proof}
By (\ref{eq:ucv}), $\vec{v}\neq c\vec{u}$ for any nonzero $c\in\mathbb{F}$. Apply Lemma~\ref{lem:randmatrix} with $\vec{b}=\vec{c}=\vec{\alpha}$, we get
\[
\Pr[f(\vec{\alpha},\vec{v})=f(\vec{\alpha},\vec{u})]=\prod_{i\in[\ell]}\Pr[\vec{\alpha}^TA_i\vec{v}=\vec{\alpha}^TA_i\vec{u}]= (1/|\mathbb{F}|)^\ell.
\]
There are at most $|V|^2$ pairs of $(\vec{v},\vec{u})$ and at most $|\mathbb{F}|^d$ choices of $\vec{\alpha}$. Since $|V|^2\cdot |\mathbb{F}|^d\cdot (1/|\mathbb{F}|)^\ell<1/10$,  with probability at least $9/10$, $f(\vec{\alpha},\vec{v})\neq f(\vec{\alpha},\vec{u})$ for all nonzero $\vec{\alpha}\in\mathbb{F}^h$ and distinct $\vec{v},\vec{u}\in V$.
\end{proof}

\begin{lemma}\label{lem:selfcor}
%Select $\ell$ matrices $A_1,A_2,\ldots,A_\ell\in \mathbb{F}^{h\times m}$ randomly and independently.  For any  nonzero vector $\vec{\alpha}\in\mathbb{F}^h$ and $\vec{v}\in V$, let  
%\[
%f(\vec{\alpha},\vec{v})=(\vec{\alpha}A_1\vec{v}^T,\cdots,\vec{\alpha} A_\ell\vec{v}^T)\in\mathbb{F}^\ell.
%\]
If $|V|^3\cdot |\mathbb{F}|^{2h}\cdot (1/|\mathbb{F}|)^\ell<1/10$, then with  probability at least $9/10$, $f(\vec{\alpha},\vec{v})+f(\vec{\alpha}',\vec{u})\neq f(\vec{\alpha}+\vec{\alpha}',\vec{w})$ for any distinct $\vec{v},\vec{u},\vec{w}\in V$ and nonzero $\vec{\alpha},\vec{\alpha}'$.
\end{lemma}
\begin{proof} Observe that $\Pr[f(\vec{\alpha},\vec{v})+f(\vec{\alpha}',\vec{u})= f(\vec{\alpha}+\vec{\alpha'},\vec{w})]= \Pr[f(\vec{\alpha},\vec{v}-\vec{w})= f(\vec{\alpha'},\vec{w}-\vec{u})]$ and $\vec{w}-\vec{u}=\vec{w}+\vec{u}\neq\vec{v}+\vec{w}=\vec{v}-\vec{w}$. 
Since $\vec{w}-\vec{u}$ and $\vec{v}-\vec{w}$ are vectors in $\{0,1\}^m\subseteq \mathbb{F}^m$, $\vec{w}-\vec{u}\neq (\vec{v}-\vec{w})$ implies $\vec{w}-\vec{u}\neq a(\vec{v}-\vec{w})$ for any $a\in\mathbb{F}\setminus\{0\}$. By Lemma~\ref{lem:randmatrix}, $\Pr[f(\vec{\alpha},\vec{v}-\vec{w})= f(\vec{\alpha'},\vec{w}-\vec{u})]\le (1/|\mathbb{F}|)^\ell$.
There are at most $|V|^3$ pairs of $(\vec{v},\vec{u},\vec{w})$ and at most $|\mathbb{F}|^{2d}$ choices of $\vec{\alpha},\vec{\alpha}'$. Since $|V|^3\cdot |\mathbb{F}|^{2d}\cdot (1/|\mathbb{F}|)^\ell<1/10$,  with probability at least $9/10$,  $f(\vec{\alpha},\vec{v})+f(\vec{\alpha}',\vec{u})\neq f(\vec{\alpha}+\vec{\alpha'},\vec{w})$ for any distinct $\vec{v},\vec{u},\vec{w}\in V$ and nonzero $\vec{\alpha},\vec{\alpha}'$.
\end{proof}

\noindent\textit{Construction of the CSP.} Let $\ell=2\log n+2h$. Then for  large $n$,
\[
(1/|\mathbb{F}|)^{\ell h}\cdot 2^{h\log n}=4^{-2h\log n-2h}\cdot 2^{h\log n}\le 2^{-3h\log n}<1/10,
\]
and
\[
|V|^3\cdot |\mathbb{F}|^{2h}\cdot (1/|\mathbb{F}|)^\ell=n^3\cdot 4^{2h}\cdot 4^{-2log n-2h}\le 1/n\le 1/10.
\]
By Lemma~\ref{lem:alphav}, Lemma~\ref{lem:selfcor} and Lemma~\ref{lem:difft}, with probability at least $7/10$, $g(\vec{v})\neq\vec{0}_{\ell h}$ for all $\vec{v}\in \mathbb{F}^{h\log n}$,  $f(\vec{\alpha},\vec{v})\neq f(\vec{\alpha},\vec{u})$ and $f(\vec{\alpha},\vec{v})+f(\vec{\alpha}',\vec{u})\neq f(\vec{\alpha}+\vec{\alpha'},\vec{w})$ for  all   distinct $\vec{v},\vec{u},\vec{w}\in V$ and nonzero $\vec{\alpha},\vec{\alpha}'$.
  
Construct a CSP instance $I$ with $|\mathbb{F}|^{kh}$ variables $\{x_{\vec{\alpha}_1,\ldots,\vec{\alpha}_k} : \vec{\alpha}_1\ldots\vec{\alpha}_k\in\mathbb{F}^h\}$. The alphabet of this CSP  is $\Gamma=\mathbb{F}^\ell$. If the instance is a yes-instance, then each $x_{\vec{\alpha}_1,\ldots,\vec{\alpha}_k}$ is expected to take the value $f(\vec{\alpha}_1,\vec{v}_1)+\cdots+f(\vec{\alpha}_k,\vec{v}_k)$ for some a solution $\vec{v}_1,\ldots,\vec{v}_k$ to the $k$-Vector-Sum problem. We now describe three types of constraints.
\begin{description}
\item[(C1)]  For all $\vec{\alpha}_1,\ldots,\vec{\alpha}_k$ and $\vec{\beta}_1,\ldots,\vec{\beta}_k$, check if $x_{\vec{\alpha}_1+\vec{\beta}_1,\ldots,\vec{\alpha}_k+\vec{\beta}_k}=x_{\vec{\alpha}_1,\ldots,\vec{\alpha}_k}+x_{\vec{\beta}_1,\ldots,\vec{\beta}_k}$.
In other words, the relation for variable tuple $(x_{\vec{\alpha}_1+\vec{\beta}_1,\ldots,\vec{\alpha}_k+\vec{\beta}_k},x_{\vec{\alpha}_1,\ldots,\vec{\alpha}_k},x_{\vec{\beta}_1,\ldots,\vec{\beta}_k})$ is
\[
R_{x_{\vec{\alpha}_1+\vec{\beta}_1,\ldots,\vec{\alpha}_k+\vec{\beta}_k},x_{\vec{\alpha}_1,\ldots,\vec{\alpha}_k},x_{\vec{\beta}_1,\ldots,\vec{\beta}_k}}=\{(\vec{a},\vec{b},\vec{c})\in\Gamma^3 : \vec{a}=\vec{b}+\vec{c}\}.
\]
\item[(C2)]  For every $i\in [k]$ and  $\vec{\alpha}_1,\ldots,\vec{\alpha}_k,\vec{\alpha}\in\mathbb{F}^h$,  check if $x_{\vec{\alpha}_1,\ldots,\vec{\alpha}_i+\vec{\alpha},\ldots,\vec{\alpha}_k}-x_{\vec{\alpha}_1,\ldots,\vec{\alpha}_k}=f(\vec{\alpha},\vec{v})$ for some $\vec{v}\in V_i$. In other words, the constraint between  $x_{\vec{\alpha}_1,\ldots,\vec{\alpha}_i+\vec{\alpha},\ldots,\vec{\alpha}_k}$ and $x_{\vec{\alpha}_1,\ldots,\vec{\alpha}_k}$ is  
\[
R_{x_{\vec{\alpha}_1,\ldots,\vec{\alpha}_i+\vec{\alpha},\ldots,\vec{\alpha}_k},x_{\vec{\alpha}_1,\ldots,\vec{\alpha}_k}}=\{(\vec{a},\vec{b})\in\Gamma^2 : \text{$\vec{a}-\vec{b}=f(\vec{\alpha},\vec{v})$ for some $\vec{v}\in V_i$}\}.
\]
\item[(C3)]  For all $\vec{\alpha}_1,\ldots,\vec{\alpha}_k,\vec{\alpha}\in\mathbb{F}^h$,  check if $x_{\vec{\alpha}_1+\vec{\alpha},\ldots,\vec{\alpha}_k+\vec{\alpha}}-x_{\vec{\alpha}_1,\ldots,\vec{\alpha}_k}=f(
\vec{\alpha},\vec{t})$. That is, 
\[
R_{x_{\vec{\alpha}_1+\vec{\alpha},\ldots,\vec{\alpha}_k+\vec{\alpha}}, x_{\vec{\alpha}_1,\ldots,\vec{\alpha}_k}}=\{(\vec{a},\vec{b})\in\Gamma^2 : \vec{a}-\vec{b}=f(\vec{\alpha},\vec{t})\}.
\]
\end{description}
Constraints of the form $x_{\vec{\alpha}_1,\ldots,\vec{\alpha}_i+\vec{\alpha},\ldots,\vec{\alpha}_k}-x_{\vec{\alpha}_1,\ldots,\vec{\alpha}_k}=f(\vec{\alpha},\vec{v})$ in (C2) are also called  (C2)-$i$-type or (C2)-$i$-$\vec{\alpha}$-type constraints.
Similarly, constraints of the form $x_{\vec{\alpha}_1+\vec{\alpha},\ldots,\vec{\alpha}_k+\vec{\alpha}}-x_{\vec{\alpha}_1,\ldots,\vec{\alpha}_k}=f(\vec{\alpha},\vec{t})$ are called  (C3)-$\vec{\alpha}$-type constraints.

%\begin{rem}
%Note that that the characteristic of $\mathbb{F}$ is two. For every $\vec{\alpha}\in\mathbb{F}^h$ and $i\in[k]$, we can partition the variables into two disjoint sets such that the constraints of (C2)-$i$-$\vec{\alpha}$-type form a matching between these sets. The same fact holds for (C3)-$\vec{\alpha}$-type constraints for every nonzero $\vec{\alpha}\in\mathbb{F}^h$.
%\end{rem}

\begin{lemma}\label{lem:kvec2csp}
If the $k$-Vector-Sum instance has a solution, then so does $I$. If the $k$-Vector-Sum instance has no solution, then there exists a constant $\epsilon>0$ such that for every assignment to the variables of $I$ one of the followings must hold:
\begin{itemize}
\item  $\epsilon/2$-fraction of constraints of (C1) are not satisfied,
\item there exists $i\in[k]$ such that  $\epsilon^2$-fraction of (C2)-$i$-type constraints  are  not satisfied. 
\item $\epsilon$-fraction of constraints of (C3) are not satisfied.
\end{itemize}
\end{lemma}
\begin{proof}
If the $k$-Vector-Sum instance has a solution $\vec{v}_1\in V_1,\ldots,\vec{v}_k\in V_k$, then let $x_{\vec{\alpha}_1,\ldots,\vec{\alpha}_k}=f(\vec{\alpha}_1,\vec{v}_1)+\cdots+f(\vec{\alpha}_k,\vec{v}_k)$. It is easy to check that  all the constraints are satisfied.

Now suppose that the $k$-Vector-Sum instance has no solution. Fix any assignment $\vec{x}\in\mathbb{F}^\ell$.
If $\epsilon/2$-fraction of constraints in (C1) are not satisfied, then we are done. Otherwise $(1-\epsilon/2)$-fraction of (C1) constraints are satisfied.
By the linearity test~\cite{blum1993self} and (C1),  there exist  $\vec{c}_1,\ldots,\vec{c}_k\in\mathbb{F}^{h\ell}$ such that for $(1-\epsilon)|\mathbb{F}^{kh}|$ many choices of $(\vec{\alpha}_1,\ldots,\vec{\alpha}_k)\in\mathbb{F}^{kh}$,  $\vec{x}_{\vec{\alpha}_1,\ldots,\vec{\alpha}_k}=F(\vec{\alpha}_1,\vec{c}_1)+F(\vec{\alpha}_2,\vec{c}_2)+\cdots+F(\vec{\alpha}_k,\vec{c}_k)$.
Let 
\[
A=\{(\vec{\alpha}_1,\ldots,\vec{\alpha}_k)\in\mathbb{F}^{kh} : \vec{x}_{\vec{\alpha}_1,\ldots,\vec{\alpha}_k}=F(\vec{\alpha}_1,\vec{c}_1)+F(\vec{\alpha}_2,\vec{c}_2)+\cdots+F(\vec{\alpha}_k,\vec{c}_k)\}.
\]
We have that $|A|\ge (1-\epsilon)|\mathbb{F}|^{kh}$.

Obviously, there are two cases:
\begin{itemize}
\item Either for every $i\in[k]$, there exists $\vec{v}_i\in V_i$ such that $\vec{c}_i=g(\vec{v}_i)$,
\item or there exists $i\in[k]$ such that $\vec{c}_i\neq g(\vec{v})$ for all $\vec{v}\in V_i$.
\end{itemize}

In the later case, we will show that  at least an $\epsilon^2$-fraction of   (C2)-$i$-type  constraints are not satisfied. 
Call  a vector $\vec{\alpha}\in\mathbb{F}^h$ good  if at least $(1-\epsilon)$-fraction of (C2)-$i$-$\vec{\alpha}$-type constraints are  satisfied. 
For every nonzero vector $\vec{\alpha}\in\mathbb{F}^h$, $\mathbb{F}^{kh}$ can be partitioned into two disjoint sets $X_{\vec{\alpha}}^-$ and $X_{\alpha}^+$ such that $X_{\vec{\alpha}}^+=\{(\vec{a}_1,\ldots,\vec{a}_i+\vec{\alpha},\ldots,\vec{a}_k) : (\vec{a}_1,\ldots,\vec{a}_k)\in X_{\vec{\alpha}}^-\}$. 
Now suppose $\vec{\alpha}$ is a good vector. There are $|\mathbb{F}^{kh}|/2$ constraints of (C2)-$i$-$\vec{\alpha}$-type. We construct a bipartite graph on $X_{\vec{\alpha}}^-$ and $X_{\vec{\alpha}}^+$. Two vertices $(\vec{a}_1,\ldots,\vec{a}_k)\in X_{\vec{\alpha}}^-$ and $(\vec{a}_1,\ldots,\vec{a}_i+\vec{\alpha},\ldots,\vec{a}_k)\in X_{\vec{\alpha}}^+$ are adjacent if
 $(\vec{x}_{\vec{a}_1,\ldots,\vec{a}_k},\vec{x}_{\vec{a}_1,\ldots,\vec{a}_i+\vec{\alpha},\ldots,\vec{a}_k})$  satisfies the constraint between them. Since $\vec{\alpha}$ is a good vector, there are $(1-\epsilon)|\mathbb{F}^{kh}|/2$ edges between $X_{\vec{\alpha}}^-$ and $X_{\vec{\alpha}}^+$. These edges form a matching $M$.
% Observe that between $X_1$ and $X_2$ are satisfied.
%For every  $\vec{\alpha}\in\mathbb{F}^h$, the variables can be partitioned into two disjoint sets $X_1,X_2$. 
Observe that $\min\{|A\cap X_1|,|A\cap X_2|\}\ge (1/2-\epsilon)|\mathbb{F}^{kh}|$. Since the size of matching $M$ is  $(1-\epsilon)|\mathbb{F}^{kh}|/2$,  when $6\epsilon<1$,  $M$ contains an edge whose endpoints are both in $A$.  In other words, there exists $(\vec{a}_1,\ldots,\vec{a}_k)\in A$ such that $(\vec{a}_1,\ldots,\vec{a}_i+\vec{\alpha},\ldots,\vec{a}_k)\in A$ and $(\vec{x}_{\vec{a}_1,\ldots,\vec{a}_i+\vec{\alpha},\ldots,\vec{a}_k}, \vec{x}_{\vec{a}_1,\ldots,\vec{a}_k})$ satisfies the (C2) constraint. By (C2) and Lemma~\ref{lem:alphav}, we deduce that $F(\vec{\alpha},\vec{c}_i)=f(\vec{\alpha},\vec{v})$ for some unique $\vec{v}\in V_i$. To summarize, for every good vector $\vec{\alpha}$, there exists a vector $\vec{v}_{\vec{\alpha}}\in V_i$ such that $F(\vec{\alpha},\vec{c}_i)=f(\vec{\alpha},\vec{v}_{\vec{\alpha}})$.

Next, we show that there are at most $(1-\epsilon) |\mathbb{F}^h|$ good vectors, and hence at most $(1-\epsilon)+\epsilon(1-\epsilon)=(1-\epsilon^2)$-fraction of constraints of   (C2)-$i$-type are satisfied. Otherwise, pick an arbitrary good vector $\vec{\alpha}$.   All the vectors in $\mathbb{F}^h-\{0,\vec{\alpha}\}$ can be partitioned into two sets $X_1$ and $X_2$ such that $X_1=X_2+\vec{\alpha}$. There exists $Y_1\subseteq X_1$  such that $Y_1$ is a set of good vectors and $|Y_1|\ge (1/2-\epsilon)|\mathbb{F}^h|$. Since $\vec{c_i}\neq g(\vec{v}_{\vec{\alpha}})$, there is a set $X$ of size at most $|\mathbb{F}|^{h-1}$ such that for all  $\vec{x}\in \mathbb{F}^h\setminus X$,  $F(\vec{x},\vec{c_i})\neq F(\vec{x},g(\vec{v}_{\vec{\alpha}}))$. Since  $(1-\epsilon)|\mathbb{F}^h|$  vectors are good, there exists a set $Z\subseteq Y_2=Y_1+\vec{\alpha}$ such that $|Z|\le\epsilon|\mathbb{F}|^h$ and all the bad vectors of $Y_2$ are in $Z$. When $2\epsilon+1/|\mathbb{F}|<1/2$, we have $|Y_1|-|X|-|Z|>0$. Thus there exists a vector $\vec{a}'\in Y_1-X-(Z+\vec{\alpha})$. According to the definitions, $\vec{\alpha}'$ and  $\vec{\alpha}+\vec{\alpha}'$ are good and $\vec{v}_{\vec{\alpha}}\neq\vec{v}_{\vec{\alpha}'}$. Since $\vec{\alpha}+\vec{\alpha}'$ is good, there exists $\vec{u}\in V_i$ such that $F(\vec{\alpha}'+\vec{\alpha},\vec{c}_i)=f(\vec{\alpha}'+\vec{\alpha},\vec{u})$. Note that  $\vec{u}\neq\vec{v}_{\vec{\alpha}}$, otherwise by $F(\vec{\alpha},\vec{c}_i)=f(\vec{\alpha},\vec{v}_{\vec{\alpha}})$, we can deduce that $F(\vec{\alpha}',\vec{c}_i)=f(\vec{\alpha}',\vec{u})$, which implies $\vec{u}=\vec{v}_{\vec{\alpha}'}\neq\vec{v}_{\vec{\alpha}}$, that is impossible. So we have $f(\vec{\alpha}'+\vec{\alpha},\vec{u})=F(\vec{\alpha},\vec{c_i})+F(\vec{\alpha}',\vec{c}_i)=f(\vec{\alpha},\vec{v}_{\vec{\alpha}})+f(\vec{\alpha}',\vec{v}_{\vec{\alpha}'})$, where $\vec{\alpha},\vec{\alpha}'$ are nonzero vectors and $\vec{v}_{\vec{\alpha}},\vec{v}_{\vec{\alpha}'},\vec{u}$ are distinct,    contradicting  Lemma~\ref{lem:selfcor}.

Now assume that $\vec{c}_i=g(\vec{v}_i)$ for every $i\in[k]$. Note that $\vec{v}_1+\cdots+\vec{v}_k\neq\vec{t}$. By Lemma~\ref{lem:difft}, $g(\sum_{i\in[k]}\vec{v}_i)\neq g(\vec{t})$. 
There exists a set $B\subseteq\mathbb{F}^h$ such that $|B|\ge (1-1/|\mathbb{F}|)\cdot|\mathbb{F}^h|$ and for all $\vec{\alpha}\in B$,
\[
\sum_{i\in [k]}F(\vec{\alpha},\vec{c}_i)=\sum_{i\in[k]}F(\vec{\alpha},g(\vec{v}_i))=F(\vec{\alpha},\sum_{i\in[k]}g(\vec{v}_i))\neq F(\vec{\alpha},g(\vec{t}))= f(\vec{\alpha},\vec{t}).
\]

Notice that $|A|\ge (1-\epsilon)|\mathbb{F}|^{kh}$. We have
\[
|\{(\vec{\alpha}_1,\ldots,\vec{\alpha}_k,\vec{\alpha}): (\vec{\alpha}_1,\ldots,\vec{\alpha}_k), (\vec{\alpha}_1+\vec{\alpha},\ldots,\vec{\alpha}_k+\vec{\alpha})\in A,\vec{\alpha}\in B\}|\ge (1-1/|\mathbb{F}|)\cdot (1-2\epsilon)\cdot |\mathbb{F}|^{k(h+1)}.
\]
This implies that at least $(1-1/|\mathbb{F}|)\cdot (1-\epsilon) |\mathbb{F}|^{k(h+1)}>\epsilon|\mathbb{F}|^{k(h+1)}$  constraints in (C3) are not satisfied.
\end{proof}

\noindent\textit{Construction of the  Gap-clique instance.} For every $\vec{\alpha}_1,\ldots,\vec{\alpha}_k\in\mathbb{F}^h$ and $\vec{\beta}_1,\ldots,\vec{\beta}_k\in\mathbb{F}^h$, introduce a vertex set $V_{\vec{\alpha}_1,\ldots,\vec{\alpha}_k,\vec{\beta}_1,\ldots,\vec{\beta}_k}=\{(\vec{x},\vec{y},\vec{z}):\vec{x},\vec{y},\vec{z}\in\mathbb{F}^\ell,  \vec{x}=\vec{y}+\vec{z}\}$. Each vertex in this set is corresponding to an assignment to three variables  $x_{\vec{\alpha}_1+\vec{\beta}_1,\ldots,\vec{\alpha}_k+\vec{\beta}_k}$, $x_{\vec{\alpha}_1,\ldots,\vec{\alpha}_k}$ and $x_{\vec{\beta}_1,\ldots,\vec{\beta}_k}$ which satisfies the constraint of (C1). For every variables $x_{\vec{\alpha}_1,\ldots,\vec{\alpha}_k}$ and $i\in [|\mathbb{F}|^{kh}]$, introduce a vertex set $V_{\vec{\alpha}_1,\ldots,\vec{\alpha}_k,i}=\mathbb{F}^\ell$. Each vertex in $V_{\vec{\alpha}_1,\ldots,\vec{\alpha}_k,i}$ is an assignment to the variable $x_{\vec{\alpha}_1,\ldots,\vec{\alpha}_k}$. 

Construct a graph $G'$ on vertice $(\bigcup_{\vec{\alpha}_1,\ldots,\vec{\alpha}_k,\vec{\beta}_1,\ldots,\vec{\beta}_k}V_{\vec{\alpha}_1,\ldots,\vec{\alpha}_k,\vec{\beta}_1,\ldots,\vec{\beta}_k})\cup (\bigcup_{\vec{\alpha}_1,\ldots,\vec{\alpha}_k,i}V_{\vec{\alpha}_1,\ldots,\vec{\alpha}_k,i})$. Two vertices in $G'$ are adjacent unless they are corresponding to inconsistent assignments or they do not satisfy the constraint between the variables they assigned to.

\begin{lemma}\label{lem:comp}
If the $k$-Vector-Sum instance has a solution, then the graph $G'$ contains a clique of size $2|\mathbb{F}|^{2kh}$.
\end{lemma}
\begin{proof}
For every vertex set, select the vertex corresponding to the assignment. According to the definition, these vertices form a clique of size $2|\mathbb{F}|^{2kh}$.
\end{proof}

\begin{lemma}\label{lem:sound}
If the $k$-Vector-Sum instance has no solution, then the graph $G'$ contains no clique of size $(1-\epsilon')2|\mathbb{F}|^{2kh}$ for some small constant $\epsilon'>0$.
\end{lemma}
\begin{proof}
Let $\epsilon$ be the constant in Lemma~\ref{lem:kvec2csp}. Pick a small $\epsilon'$ such that $4\epsilon'<\min\{\epsilon/2,\epsilon^2\}$.
Let $X$ be the clique in the graph of size larger than $(1-\epsilon')2|\mathbb{F}|^{2kh}$. We have  
\begin{equation}\label{eq:c3X}
|X\cap \bigcup_{\vec{\alpha}_1,\ldots,\vec{\alpha}_k,\vec{\beta}_1,\ldots,\vec{\beta}_k\in\mathbb{F}^h}V_{\vec{\alpha}_1,\ldots,\vec{\alpha}_k,\vec{\beta}_1,\ldots,\vec{\beta}_k}|\ge (1-2\epsilon')|\mathbb{F}|^{2kh}.
\end{equation} 
%in $V_{\vec{\alpha}_1,\ldots,\vec{\alpha}_k,\vec{\beta}_1,\ldots,\vec{\beta}_k}$.
In addition, since $|X|> (1-\epsilon')2|\mathbb{F}|^{2kh}$,  there exists an index $i\in[|\mathbb{F}^{kh}|]$ such that $X$ contains more than  $(1-2\epsilon')|\mathbb{F}^{kh}|$ vertices with respect to index $i$. We will prove that this is impossible using the following Claim 1.

\medskip
 
\noindent\textit{Claim 1.}  For every $i\in[|\mathbb{F}^{kh}|]$, $|X\cap \bigcup_{\vec{\alpha}_1,\ldots,\vec{\alpha}_k\in\mathbb{F}^h}V_{\vec{\alpha}_1,\ldots,\vec{\alpha}_k,i}|\le (1-2\epsilon')|\mathbb{F}^{kh}|$.

\medskip

\noindent\textit{Proof of Claim 1.}
%Assume that $|X\cap \bigcup_{\vec{\alpha}_1,\ldots,\vec{\alpha}_k\in\mathbb{F}^h}V_{\vec{\alpha}_1,\ldots,\vec{\alpha}_k,i}|> (1-2\epsilon')|\mathbb{F}^h|$ for some $i\in[|\mathbb{F}|^{kh}]$. 
Fix an index $i\in[|\mathbb{F}|^{kh}]$. Define an assignment $\sigma_X$ as follows. For every  $\vec{\alpha}_1,\ldots,\vec{\alpha}_k\in\mathbb{F}^h$, $\sigma_X(x_{\vec{\alpha}_1,\ldots,\vec{\alpha}_k})=v$ if $\{v\}= X\cap V_{\vec{\alpha}_1,\ldots,\vec{\alpha}_k,i}$, otherwise $\sigma_X(x_{\vec{\alpha}_1,\ldots,\vec{\alpha}_k})=\vec{0}_\ell$.
By (\ref{eq:c3X}) and the definition of edge set, $\delta_X$ satisfies $(1-2\epsilon')$-fraction of constraints in (C1). By Lemma~\ref{lem:kvec2csp}, either $\sigma_X$ satisfies at most $(1-4\epsilon')$-fraction of constraints in (C3) or there exists an $j\in[k]$ such that $\sigma_X$ satisfies at most $(1-4\epsilon')$-fraction of constraints of (C2)-$j$ type.
% for index $i$.

\begin{itemize}
\item  Assume that  $\sigma_X$ satisfies at most $(1-4\epsilon')$-fraction of constraints in (C3).
We say a vector $\vec{\alpha}\in \mathbb{F}^h$ is $\delta$-good if more than $\delta$-fraction of (C3)-$\vec{\alpha}$-type constraints  are satisfied by $\sigma_X$. There exists a vector $\vec{\alpha}\in \mathbb{F}^h$ that is not $(1-4\epsilon')$-good, otherwise $\delta_X$ satisfies more than $(1-4\epsilon')$-fraction of constraints in (C3), contradicting  our assumption. Now consider a partition $\mathbb{F}^{kh}=V_{\vec{\alpha}}^-\cup V_{\vec{\alpha}}^+$ with $V_{\vec{\alpha}}^+=V_{\vec{\alpha}}^-+(\vec{\alpha},\ldots,\vec{\alpha})$.
Since $\vec{\alpha}$ is not $(1-4\epsilon')$-good, there are at most $(1-4\epsilon') |\mathbb{F}^{kh}|/2$ tuples $(\vec{\alpha}_1,\ldots,\vec{\alpha}_k)$ in $V_{\vec{\alpha}}^-$ such that $\delta_X$ satisfies the (C3) constraint between $x_{\vec{\alpha}_1,\ldots,\vec{\alpha}_k}$ and $x_{\vec{\alpha}_1+\vec{\alpha},\ldots,\vec{\alpha}_k+\vec{\alpha}}$.
Let 
\[
X_{\vec{\alpha}}^-=X\cap\bigcup_{(\vec{\alpha}_1,\ldots,\vec{\alpha}_k)\in V_{\vec{\alpha}}^-}V_{\vec{\alpha}_1,\ldots,\vec{\alpha}_k,i}
\text{ and }
X_{\vec{\alpha}}^+=X\cap\bigcup_{(\vec{\alpha}_1,\ldots,\vec{\alpha}_k)\in V_{\vec{\alpha}}^+}V_{\vec{\alpha}_1,\ldots,\vec{\alpha}_k,i}
\]
It follows that $\min\{|X_{\vec{\alpha}}^-|,|X_{\vec{\alpha}}^+|\}\le (1-4\epsilon') |\mathbb{F}|^h/2$. Thus $|X|\le (1/2+(1-4\epsilon')/2)|\mathbb{F}^{kh}|=(1-2\epsilon')|\mathbb{F}|^{kh}$.

\item Now assume that  $\sigma_X$ satisfies at most $(1-4\epsilon')$-fraction of constraints of (C2)-$j$ type for some $j\in[k]$.
% for index $i$.
Similarly, for every  $\vec{\alpha}\in \mathbb{F}^h$, we say $\vec{\alpha}$ is $\delta$-good if more than $\delta$-fraction of (C2)-$j$-$\vec{\alpha}$-type constraints are satisfied by $\sigma_X$.
There exists  $\vec{\alpha}\in \mathbb{F}^h$ that is not $(1-4\epsilon')$-good, otherwise $\delta_X$ satisfies more than $(1-4\epsilon')$-fraction of constraints of type (C2)-$j$, contradicting  our assumption. Now consider a partition $\mathbb{F}^{kh}=V_{\vec{\alpha}}^-\cup V_{\vec{\alpha}}^+$ with \[
V_{\vec{\alpha}}^+=\{(\vec{\alpha}_1,\ldots,\vec{\alpha}_i+\vec{\alpha},\ldots,\vec{\alpha}_k): (\vec{\alpha}_1,\ldots,\vec{\alpha}_k)\in V_{\vec{\alpha}}^-\}.
\]
Let 
\[
X_{\vec{\alpha}}^-=X\cap\bigcup_{(\vec{\alpha}_1,\ldots,\vec{\alpha}_k)\in V_{\vec{\alpha}}^-}V_{\vec{\alpha}_1,\ldots,\vec{\alpha}_k,i}
\text{ and }
X_{\vec{\alpha}}^+=X\cap\bigcup_{(\vec{\alpha}_1,\ldots,\vec{\alpha}_k)\in V_{\vec{\alpha}}^+}V_{\vec{\alpha}_1,\ldots,\vec{\alpha}_k,i}.
\]

Since $\vec{\alpha}$ is not $(1-4\epsilon')$-good, we have  $\min\{|X_{\vec{\alpha}}^-|,|X_{\vec{\alpha}}^+|\}\le (1-4\epsilon') |\mathbb{F}|^h/2$. Thus  $|X|\le (1/2+(1-4\epsilon')/2)|\mathbb{F}^{kh}|=(1-2\epsilon')|\mathbb{F}|^{kh}$.
\end{itemize}
\end{proof}

\subsection{Putting all together}
For any instance $(G,k)$ of $k$-Clique, we use Theorem~\ref{thm:kvec} to reduce it to an instance $(k',V_1,\ldots,V_{k'},\vec{t})$ of $k'$-Vector-Sum with $k'=\Theta(k^2)$. Then we use the reduction describe to obtain a graph $G'$ with $h=k'^2=\Theta(k^4)$ and small $\epsilon>0$ in $2^{k^{O(1)}}\cdot |G|^{O(1)}$-time. By Lemma~\ref{lem:comp} and Lemma~\ref{lem:sound}, we have
\begin{itemize}
\item if $\omega(G)\ge k$, then $\omega(G')\ge 2^{4kh+1}$,
\item if $\omega(G)<k$, then with probability at least $7/10$, $\omega(G')< (1-\epsilon)2^{4kh+1}$.
\end{itemize}
Using the graph product method, we can amplify the gap to any constant.

\subsection{Derandomization}
To derandomize the reduction, we need to construct $O(\log n+h)$ matrices $A_1,\ldots,A_\ell\in\mathbb{F}^{h\times m}$ such that the following conditions are satisfied. For all nonzero $\vec{v}\in\mathbb{F}^m=\mathbb{F}^{h\log n}$,
\begin{equation}\label{eq:difft}
(A_1\vec{v},\ldots,A_\ell\vec{v})\neq\vec{0}_{\ell h}
\end{equation}
For all distinct $\vec{v},\vec{u}\in V$ and nonzero $\vec{\alpha}\in\mathbb{F}^h$,
\begin{equation}\label{eq:alphav}
(\vec{\alpha}^T A_1\vec{v},\ldots,\vec{\alpha}^T A_\ell\vec{v})\neq (\vec{\alpha}^T A_1\vec{u},\ldots,\vec{\alpha}^T A_\ell\vec{u}) 
\end{equation}
For all distinct $\vec{v},\vec{u},\vec{w}\in V$ and nonzero $\vec{\alpha},\vec{\alpha'}\in\mathbb{F}^h$
\begin{equation}\label{eq:selfcord}
(\vec{\alpha}^T A_1(\vec{v}+\vec{w}),\ldots,\vec{\alpha}^T A_\ell(\vec{v}+\vec{w}))\neq (\vec{\alpha'}^T A_1(\vec{u}+\vec{w}),\ldots,\vec{\alpha'}^T A_\ell(\vec{u}+\vec{w})) 
\end{equation}
Let $A_i\in\mathbb{F}^{h\times m}$ be the matrix such that $A_i\vec{v}$ is the projection of $\vec{v}$ onto  the its subvector  with coordinates between $1+(i-1)h$ and $ih$. Then $A_1,\ldots,A_{\log n}$ satisfy (\ref{eq:difft}). It remains to construct another $O(\log n+h)$ matrices $A_1',\ldots,A_{\ell}'\in\mathbb{F}^{h\times m}$ satisfying (\ref{eq:alphav}) and (\ref{eq:selfcord}). Then their union $A_1,\ldots,A_{\log n},A_1',\ldots,A_{\ell}'$ would satisfy all the conditions. Note that we can think of a matrix in $\mathbb{F}^{h\times m}$ as an $hm$-dimension vector. The task   can be formulated as given $N= |\mathbb{F}|^{2h}\cdot n^{O(1)}$ vectors $C_1,\ldots,C_N\in \mathbb{F}^{hm}$, find $O(\log n+h)$ vectors $A_1',\ldots,A_{\ell}'\in\mathbb{F}^{hm}$ such that for every $i\in[N]$, there exists  $A_j'$ such that $A_j'\cdot C_i\neq 0$.

We show that, in $N^{O(1)}$-time, a vector $A\in\mathbb{F}^{hm}$ can be found such that there are at most $N/|\mathbb{F}|$ indices $i\in[N]$ satisfying $A\cdot C_i=0$. Then we apply this algorithm $\log N/\log |\mathbb{F}|$ times to obtain the vectors $A_1',\ldots,A_{\ell}'$. The vector $A$ can be found using the method of conditional probabilities~\cite{jukna2011extremal,alon2004probabilistic}. 
Let $A$ be a  vector with $A[i]$ selected randomly and independently from $\mathbb{F}$. Define a random variable $X=|\{i\in[N]:A\cdot C_i=0\}|$. We have $E[X]=N/|\mathbb{F}|$. For $a_1,\ldots,a_i\in\mathbb{F}$, let $X|a_1,\ldots,a_i= |\{i\in[N]:A\cdot C_i=0,A[1]=a_1,\ldots,A[i]=a_i\}|$. We have
\[
E[X|a_1,\ldots,a_i]=\sum_{x\in\mathbb{F}}E[X|a_1,\ldots,a_i,x]/|\mathbb{F}|\ge \min\{E[X|a_1,\ldots,a_i,x]:x\in\mathbb{F}\}.
\]
For each $i\in[hm]$, $E[X|a_1,\ldots,a_i]$ can be computed in $N^{O(1)}$-time. For each $i\in[hm]$, we pick the value $a_i$ to minimize  $E[X|a_1,\ldots,a_i]$. We have $E[X|a_1,\ldots,a_{hm}]\le N/|\mathbb{F}|$ and the vector $A$ with $A[i]=a_i$($i\in[hm]$) is our target vector.
% the corresponding functions $g$ and $f$ satisfy the properties in Lemma~\ref{lem:difft}, Lemma~\ref{lem:alphav} and Lemma~\ref{lem:selfcor}.
%Use the method of conditional probabilities.
%Let $X(A_1,\ldots,A_i)=|\{g(\vec{v})=\vec{0}_{\ell h}\}|$. Pick $A_i$ that maximize the function.
%$E[X]\le 1/10$.

%or $(\vec{\theta}_1,\ldots,\vec{\theta}_k)$ is not equal to $(\vec{\alpha}_1,\ldots,\vec{\alpha}_k)$, $(\vec{\beta}_1,\ldots,\vec{\beta}_k)$ nor $(\vec{\alpha}_1+\vec{\beta}_1,\ldots,\vec{\alpha}_k+\vec{\beta}_k)$.
\section{Conclusion}
This paper constructs a PCP verifier  which always accepts  yes-instances and with probability $\Theta(1/k)$ rejects  no-instances for a $W[1]$-hard problem and  shows how to create a constant  gap for $\omega(G)$ using  this PCP. I hope that the technique of this paper will help  obtain a parameterized version of PCP theorem, e.g. the  Parameterized Inapproximability Hypothesis (PIH)~\cite{pih}. 

% A candidate is the Parameterized Inapproximability Hypothesis (PIH) in ~\cite{pih}. Can we prove  the FPT-inapproximability of $k$-Clique implies PIH?

% There are several interesting questions for future work: (1) Can we prove  the FPT-inapproximability of other parameterized problems using similar technique? (2) Show the FPT-inapproximability of $k$-Clique implies PIH. 

%(3)  Derandomize the reduction of this paper.
%\begin{itemize}
%\item Can we prove FPT-inapproximability of other parameterized problems using this PCP?
%\item Can we prove $o(k)$-FPT inapproximability of $k$-Clique assuming $W[1]\neq FPT$?
%\item Can we prove FPT-inapproximability of $k$-Clique implies PIH?
%\end{itemize}

%Recently, Whether there exists an $f(k)\cdot|G|^{O(1)}$-time algorithm that can
%lead to better understanding of PIH.

\bibliography{ref}

\section*{Appendix}
\begin{theo}[Theorem~\ref{thm:kvec} restated]
$k$-Vector-Sum with  $\mathbb{F}=\mathbb{F}_2$   and $m=\Theta(k^2\log n)$ is $W[1]$-hard parameterized by $k$.
\end{theo}

\begin{proof}
We construct a reduction from $k$-Multi-Color-Clique to $(k+\binom{k}{2})$-Vector-Sum. Let $(G,k)$ be an instance of $k$-Multi-Color-Clique with $V(G)=V_1\cup V_2\cup\cdots\cup V_k$. Set $n=|V(G)|+1$.
For every $v\in V(G)$, let $\sigma(v)\in\{0,1\}^{\log n}$ be the binary encoding of $v$. Since $n>V(G)$, we can assume that $\sigma(v)\neq\vec{0}_{\log n}$ for all $v\in V(G)$. Let $f:\binom{[k]}{2}\to [\binom{k}{2}]$ be a bijection.
For every $i\in[k]$ and $v\in V(G)$, let
\[
\vec{\eta}_{v,i}=(\underbrace{\overbrace{0,\cdots,0}^{(i-1)\log n},\sigma(v),0,\cdots,0)}_{k\log n}\text{ and }
\vec{\gamma}_{v,i}=(\underbrace{\overbrace{\sigma(v),\ldots,\sigma(v)}^{(i-1)\log n},\vec{0}_{\log n},\sigma(v),\ldots,\sigma(v)}_{k\log n})
\]
For any distinct $i,j\in[k]$, let
\[
\vec{\theta}_{i,j}=(\underbrace{\overbrace{0,\cdots,0}^{f(\{i,j\})-1},1,0,\cdots,0)}_{k(k-1)/2}\text{ and } 
\vec{\delta}_{i}=(\underbrace{\overbrace{0,\cdots,0}^{i-1},1,0,\cdots,0)}_{k}.
\]
For every edge $e=\{v,u\}$ with $v\in V_i$ and $u\in V_j$, let
\[
\vec{w}_e=\vec{0}_k\circ\vec{\theta}_{i,j}\circ (\underbrace{\overbrace{\overbrace{0,\ldots,0}^{(i-1)k\log n},\vec{\eta}_{v,j},0,\ldots,0}^{(j-1)k\log n},\vec{\eta}_{u,i},0,\ldots,0}_{k^2\log n}).
\]
For every $v\in V_i$, let
\[
\vec{w}_v=\vec{\delta}_i\circ\vec{0}_{k(k-1)/2}\circ (\underbrace{\overbrace{0,\ldots,0}^{(i-1)k\log n},\vec{\gamma}_{v,i},0,\ldots,0}_{k^2\log n})
\]
The instance of vector sum is defined as follows.
\begin{itemize}
\item The target vector is $\vec{t}=\vec{1}_{k+k(k-1)/2}\circ\vec{0}_{k^2\log n}$.
\item There are $k(k-1)/2+k$ sets of vectors. 
\begin{itemize}
\item For every $\{i,j\}\in\binom{[k]}{2}$, let
\[
W_{ij}=\{\vec{w}_e : \text{$e=\{v,u\}$ is an edge in $G$ with $v\in V_i$ and $u\in V_j$}\}.
\]
\item For every $i\in[k]$, let
\[
W_i=\{\vec{w}_v : v\in V_i\}.
\]
\end{itemize}
\end{itemize}
If $(G,k)$ is a yes-instance, then there exist $v_1\in V_1,\ldots,v_k\in  V_k$ such that $\{v_1,\ldots,v_k\}$ induces a $k$-clique in $G$. It is easy to check that the sum of $\vec{w}_{v_iv_j}$'s and $\vec{w}_{v_i}$'s is equal to $\vec{t}$.

On the other hand, if there exist $\vec{w}_{ij}\in W_{ij}$ and $\vec{w}_i\in W_i$ such that
\[
\sum \vec{w}_{ij}+\sum \vec{w}_i=\vec{t}.
\]
Each $w_i$ is corresponding to a vertex $v_i\in V_i$. 
Each $w_{ij}$ is corresponding to an edge $e_{ij}$ between $V_i$ and $V_j$.
It is easy to see that $v_i$ is an endpoint of $e_{ij}$ for all $j\in [k]\setminus\{i\}$. Therefore they form a clique of size $k$.
% Otherwise the vector $\vec{\eta}_{v_i}$ in $\vec{w}_{v_i}$ is not canceled.
%Each $\vec{\eta}_i$ must be canceled. Therefore it has at least $k-1$ neighbors.
\end{proof}
\end{document}